\documentclass[conference]{IEEEtran}
\usepackage{cite}
\usepackage{amsmath,amssymb,amsfonts}
\usepackage{algorithmic}
\usepackage{graphicx}
\usepackage{textcomp}
\usepackage{xcolor}
\usepackage{hyperref}
\usepackage{amsthm}
\usepackage[section]{placeins}
\newtheorem{proposition}{Proposition}
\def\BibTeX{{\rm B\kern-.05em{\sc i\kern-.025em b}\kern-.08em
    T\kern-.1667em\lower.7ex\hbox{E}\kern-.125emX}}
\IEEEoverridecommandlockouts
\begin{document}
\title{Channel estimation for double IRS assisted broadband single-user SISO communication
\thanks{Vishnu Karthikeya Gorty is with the Department of Electrical Engineering, Indian Institute of Technology Delhi, INDIA-110016 (E-mail: vishnukgorty@ee.iitd.ac.in).}
}
\author{Vishnu Karthikeya Gorty}
\maketitle
\begin{abstract}
In this paper, two Intelligent reflecting surfaces (double IRS) assisted single-user single input single output (SISO) communication system is considered. The cascaded channels (mobile user (MU)$\rightarrow$IRS-1$\rightarrow$base station (BS), MU$\rightarrow$IRS-2$\rightarrow$BS and MU$\rightarrow$IRS-1$\rightarrow$IRS-2$\rightarrow$BS channels) are estimated under Bayesian setting. Here, the goal is to evaluate the performance of the estimator in case of MU$\rightarrow$IRS-1$\rightarrow$BS and MU$\rightarrow$IRS-2$\rightarrow$BS channel links using Bayesian Cram{\'e}r-Rao lower bound (CRLB). Without the knowledge of closed form pdf of inner product of circularly symmetric complex Gaussian (CSCG) random vectors, we cannot obtain the fisher information. Hence, by numerical computation we obtain the Bayesian CRLB. In the simulation results, we show that we can approximate the pdf of the inner product of CSCG random vectors by a Rayleigh distribution by increasing the number of elements on the IRS, which is analogous to Central Limit Theorem (CLT). Also, the results convey that the mean squared error (MSE) almost matches with the Bayesian CRLB.
\end{abstract}

\begin{IEEEkeywords}
Intelligent Reflecting Surface (IRS), Channel Estimation (CE)
\end{IEEEkeywords}

\section{Introduction}
Intelligent reflecting surface (IRS), a low cost and low energy consuming technology, is made up of a planar array of passive elements which can reflect, refract, attenuate and induce a desired phase shift in the impinging signals which can help in constructive or destructive interference of the impinging signals. To utilize the performance gains offered by this technology, there is a need to design the passive beamforming coefficients at the IRS and active beamforming coefficents at the BS. This poses a challenge to obtain the accurate channel estimates as the IRS do not have any RF chains to transmit the pilots. Hence, in this paper, we takeup the problem of cascaded channel estimation in case of double IRS assisted broadband single-user system. Here, we compute the linear minimum mean squared error (LMMSE) estimate of the cascaded channels. To measure the performance of the estimator (in case of MU$\rightarrow$IRS-1$\rightarrow$BS and MU$\rightarrow$IRS-2$\rightarrow$BS channel links), we need to compute the Bayesian CRLB. This leads to the problem of computing the fisher information, which is difficult because we do not know the closed form pdf of the inner product of CSCG random vectors. Hence, by resorting to numerical computation, we show that this pdf can be approximated by a Rayleigh distribution.
\subsection{Related work}
In the recent years, there have been a significant research work related to IRS channel estimation. These works can be grouped into two categories: CE for single-IRS assisted communication and CE for double-IRS assisted communication.

\paragraph{CE for single-IRS assisted communication}: In \cite{b5}, a narrowband single user multiple input single output (MISO) system is considered. With a time division duplex (TDD) system assumption, the MMSE estimate of the direct channel, BS$\rightarrow$MU, is computed by switching OFF the IRS. Next, the cascaded channel associated with each IRS element is estimated.
In \cite{b6}, an IRS assisted multi-user communication is considered wherein a three stage channel estimation protocol is proposed to estimate the direct channel and the cascaded channels of all the users. The idea is to estimate the channels of all other users under consideration based on the channel knowledge of a reference user.
In \cite{b7}, a semi-passive IRS is considered to estimate the BS$\rightarrow$IRS and MU$\rightarrow$IRS channel links using the orthogonal matching pursuit (OMP) algorithm.
In \cite{b8}, a certain number of IRS elements are grouped together, where each group is called a sub-surface. Based on this grouping method, in \cite{b9}, the cascaded channel is estimated in a orthogonal frequency division multiplexing (OFDM) system.

\paragraph{CE for double-IRS assisted communication}: There may arise a scenario where the signal transmitted from the MU reaches the BS by reflection from two IRS's. So, \cite{b10} addresses the problem of channel estimation in case of double-IRS. Here, it is assumed that the MU$\rightarrow$IRS-2 and IRS-1$\rightarrow$BS channels are assumed to be blocked and the cascaded channel, MU$\rightarrow$IRS-1$\rightarrow$IRS-2$\rightarrow$BS is estimated. In \cite{b11}, a double-IRS assisted multi-user multiple input single output (MIMO) case is considered where only the direct channel, MU$\rightarrow$BS is blocked unlike in \cite{b10}. Here, all the cascaded channels are estimated with an effort to lower the pilot overhead.

\subsection{Contributions}
From the above discussion, we can say that there have been a lot of works that address the problem of channel estimation for IRS assisted communication. However, to the best of our knowledge, these works did not consider the computation of CRLB to evaluate the performance of the estimator. Hence, we make the following contributions in this paper:
\begin{itemize}

\item The performance of the estimator (in case of MU$\rightarrow$IRS-1$\rightarrow$BS and MU$\rightarrow$IRS-2$\rightarrow$BS channel links) is evaluated using an alternative tool, Bayesian CRLB, unlike the benchmark schemes used in the literature.

\item To compute the CRLB, we need the fisher information which in turn requires the closed form expression of the pdf of inner product of CSCG random vectors. In the simulation results, we show that this problem can be solved numerically by generating a large number of samples from the inner product and thereby approximating the pdf of the inner product by a Rayleigh distribution.
\end{itemize}

\subsection{Notations}A matrix or a vector is denoted by $\mathbf{A}$. The conjugate transpose and transpose of $\mathbf{A}$ is denoted by $\mathbf{A}^H$ and $\mathbf{A}^T$ respectively.
Let $\mathbb{E}\{\cdot\}$ and $Var\{\cdot\}$ denote the statistical expectation and variance respectively.
\subsection{Organization of the paper}
In Section \ref{system model}, the assumptions, channel statistics for Bayesian CE, transmission scheme and the signal model are discussed. In Section \ref{CE, MSE and CRLB}, the calculation of channel estimates, MSE and CRLB in case of Bayesian setting is discussed. The simulation results are presented in Section \ref{sim results}. Section \ref{conclusion} concludes the paper.

\section{System model}\label{system model}
We consider a double IRS assisted broadband communication between a MU and BS as shown in Figure \ref{fig1}. We assume that the MU and BS are equipped with a single antenna. Also, we assume that the direct path between the BS and the user is blocked due to a building or any other obstacle. Further, we assume that there is no LoS between any entities, that is, the BS, IRS's and the MU. Let $N$ denote the number of sub-carriers in an OFDM symbol and the number of pilots are equal to the number of sub-carriers in an OFDM symbol. The number of reflecting elements on IRS-1 and IRS-2 be $N_1$ and $N_2$ respectively.
\begin{figure}[!htb]
    \centering
    \includegraphics[height=6cm,width = 8cm,]{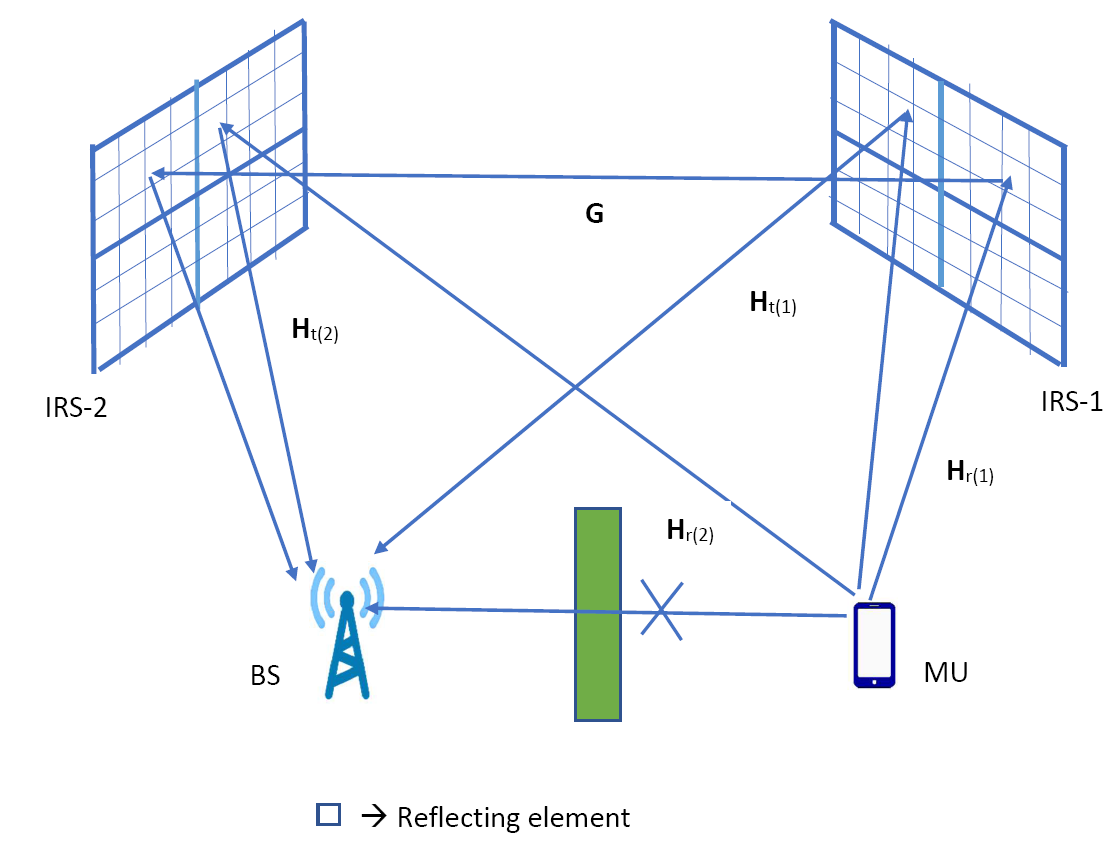}
    \caption{Schematic diagram of double IRS assisted uplink communication}
    \label{fig1}
\end{figure}

Let all the channels between the MU and BS be $L$-tap channels. After padding $\{N-L\}$ number of zeros in the channel impulse response (CIR) of each channel, the channel frequency response (CFR) of each channel is defined as follows. Let $\mathbf{H}_{r(1),k} \in \mathbb{C}^{N_1\times 1}$ and $\mathbf{H}_{r(2),k} \in \mathbb{C}^{N_2 \times 1}$ denote the CFR between the MU to IRS-1 and IRS-2 respectively over the $k^{th}$ sub-carrier in an OFDM symbol. Similarly, let $\mathbf{H}_{t(1),k} \in \mathbb{C}^{N_1\times 1 }$ and $\mathbf{H}_{t(2),k}  \in \mathbb{C}^{N_2\times 1}$ denote the channels, IRS-1$\rightarrow$BS and IRS-2$\rightarrow$BS respectively over the $k^{th}$ sub-carrier in an OFDM symbol. The CFR of IRS-1$\rightarrow$IRS-2 channel is denoted by $\mathbf{G}_{k}\in \mathbb{C}^{N_2\times N_1}$ over a sub-carrier in an OFDM symbol. Let $H_{t(i),k,n_i}, H_{r(i),k,n_i}$ and $G_{k,n_1,n_2}$ denote the elements in $\mathbf{H}_{t(i),k}, \mathbf{H}_{r(i),k}$ and $\mathbf{G}_k$ respectively, where $i\in\{1,2\},n_1\in\{1,2,...,N_1\}, n_2\in\{1,2,..., N_2\}$. All the channels are assumed to be following a quasi-static block fading channel model and the system is assumed to be operating in TDD mode.

Let $V_{(m),k} \sim \mathcal{CN}(0,\sigma^2)$ denote the frequency domain noise sample over the $k^{th}$ sub-carrier and they are independent and identically distributed (i.i.d) ,where $m \in \{1,2,3\}$ and $k \in \{1,2,..., N\}$. Let $J_{(1),k} = \mathbf{H}_{t(1),k}^T \mathbf{H}_{r(1),k}$ denote the cascaded channel associated with IRS-1. Similarly, let $J_{(2),k} = \mathbf{H}_{t(2),k}^T \mathbf{H}_{r(2),k}$ and $J_{(3),k} = \mathbf{H}_{t(2),k}^T\mathbf{G}_{k} \mathbf{H}_{r(1),k}$ denote the cascaded channels MU$\rightarrow$IRS-2$\rightarrow$BS and MU$\rightarrow$IRS-1$\rightarrow$IRS-2$\rightarrow$BS over the $k^{th}$ sub-carrier respectively.
Suppose that $H_{t(1),k,n_1}\sim \mathcal{CN}(0,\sigma_2^2)$, $H_{r(1),k,n_1}\sim \mathcal{CN}(0,\sigma_1^2)$, $H_{r(2),k,n_2}\sim \mathcal{CN}(0,\sigma_3^2)$, $H_{t(2),k,n_2}\sim \mathcal{CN}(0,\sigma_4^2)$  and $G_{k,n_1,n_2}\sim\mathcal{CN}(0,\sigma_5^2)$. Then $\mathbb{E}\{J_{(1),k}\} = \mathbb{E}\{J_{(2),k}\} = \mathbb{E}\{J_{(3),k}\} = 0 $ and $Var\{J_{(1),k}\} = N_1\sigma_1^2\sigma_2^2, Var\{J_{(2),k}\} = N_2\sigma_3^2\sigma_4^2 $ and $Var\{J_{(3),k}\}= N_2N_1\sigma_1^2\sigma_4^2\sigma_5^2 $.

The reflection coefficient of an IRS element is given by, $$\phi_n = \beta_n e^{j\theta_n}$$ where '$n$' denotes the $n^{th}$ reflecting element on the IRS. Since we focus on CE in this paper, we set the reflection amplitude, $\beta_n = 1$ and assume that the reflection coefficient of all the elements on IRS-1 and IRS-2 is $\phi_1$ and $\phi_2$ respectively.
\subsection{Transmission scheme}
\label{transmission scheme in SISO}
The quasi-static block fading assumption means that the cascaded channels remain constant over a time block and varies from one time block to another. Hence, in this paper, we focus on a single time block, where a part of it is used to estimate the cascaded channels and the remaining time is used for data transmission as shown in Figure-\ref{fig2}. The magnified version of each CE block in Figure-\ref{fig2} is shown in Figure-\ref{fig2_ext}.
\begin{figure}[htb]
	 \centering
    \includegraphics[height = 4cm,width = 6cm]{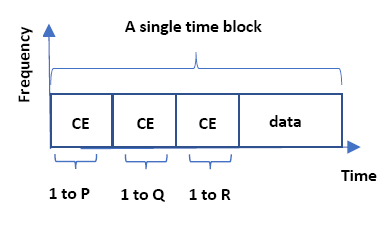}
    \caption{A single time block showing CE and data transmission phase}
    \label{fig2}
\end{figure}
\begin{figure}[htb]
    \centering
    \includegraphics[height = 4cm,width = 6cm]{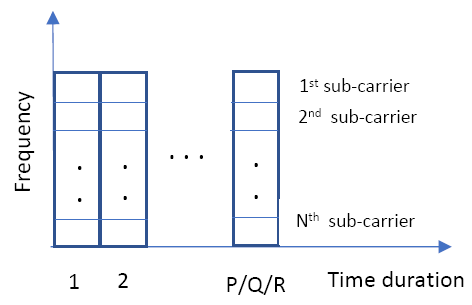}
    \caption{Transmission scheme of OFDM symbols}
    \label{fig2_ext}
\end{figure}
To estimate the MU$\rightarrow$IRS-1$\rightarrow$BS channel,
the MU transmits '$P$' pilot symbols over the $k^{th}$ sub-carrier and each symbol is denoted by  $X_{(1),k}^{(p)}$, where $p\in \{1,2,...,P\}$.
Similarly, the MU transmits $X_{(2),k}^{(q)} $ and $X_{(3),k}^{(r)}$ pilot symbols to estimate the channels 
 MU$\rightarrow$IRS-2$\rightarrow$BS and MU$\rightarrow$IRS-1$\rightarrow$IRS-2$\rightarrow$BS respectively. Here, $q\in \{1,2,...,Q\}$ and $r\in \{1,2,...,R\}$.
\subsection{Signal model}
From the above discussion, consider the case of CE of MU$\rightarrow$IRS-2$\rightarrow$BS channel. We can write the received signal at the BS over the $k^{th}$ sub-carrier as
\begin{equation}
       Y_{(2),k}^{(q)} = X_{(2),k}^{(q)} \phi_2 (\mathbf{H}_{t(2),k}^T \mathbf{H}_{r(2),k})+V_{(2),k}^{(q)}
\end{equation}
\begin{proposition}\label{prop:mean and variance of J}
The $\mathbb{E}\{J_{(1),k}\} = 0$ and $Var\{J_{(1),k}\} = N_1\sigma_1^2\sigma_2^2$.
\end{proposition}
\begin{proof}
See Appendix-\ref{appendix:proof of prop 1}
\end{proof}
\section{Computation of estimates, MSE and CRLB  }\label{CE, MSE and CRLB}
 As shown in Figure \ref{fig2}, in the first CE block, the MU$\rightarrow$IRS-1$\rightarrow$BS channel is estimated by switching OFF all the elements of IRS-2 and switching ON all the elements of IRS-1. Similarly, in the second CE block, all the elements of IRS-2 are switched ON and that of IRS-1 are switched OFF. In the third CE block, all the elements of IRS-1 and IRS-2 are switched ON.
\paragraph{Estimation of MU$\rightarrow$IRS-1$\rightarrow$BS channel}
By stacking the pilots sent over $k^{th}$ sub-carrier for '$P$' time slots, we get
\begin{equation}
	 \mathbf{Y}_{(1),k}= \phi_1\mathbf{X}_{(1),k} J_{(1),k} +\mathbf{V}_{(1),k} 
	\label{e10}
\end{equation}
where $\mathbf{X}_{(1),k} \in \mathbb{C}^{P\times 1}$. Then the LMMSE estimate is given by
\begin{equation}
	\hat{J}_{(1),k} = \mathbf{R}_{JY}\mathbf{R}_{YY}^{-1}\mathbf{Y}_{(1),k}
	\label{e11}
\end{equation}
where $\mathbf{R}_{JY} =\dfrac{N_1\sigma_1^2\sigma_2^2}{\phi_1}\mathbf{X}_{(1),k}^H$ is the cross-covariance of $J_{(1),k}$ and $\mathbf{Y}_{(1),k}$. $\mathbf{R}_{YY} = N_1\sigma_1^2\sigma_2^2(\mathbf{X}_{(1),k}\mathbf{X}_{(1),k}^H) + \sigma^2\mathbf{I}$ is the covariance of $\mathbf{Y}_{(1),k}$.\\
From \eqref{e11}, the LMMSE estimate is given by 
\begin{equation}
	\hat{J}_{(1),k} = \dfrac{N_1\sigma_1^2\sigma_2^2}{\phi_1}\mathbf{X}_{(1),k}^H(N_1\sigma_1^2\sigma_2^2(\mathbf{X}_{(1),k}\mathbf{X}_{(1),k}^H) + \sigma^2\mathbf{I})^{-1}\mathbf{Y}_{(1),k}
	\label{e12}
\end{equation}

\paragraph{Estimation of MU$\rightarrow$IRS-2$\rightarrow$BS channel}
The LMMSE estimate is given by
\begin{multline}
	\hat{J}_{(2),k} =\\
	 \dfrac{N_2\sigma_3^2\sigma_4^2}{\phi_2}\mathbf{X}_{(2),k}^H(N_2\sigma_3^2\sigma_4^2(\mathbf{X}_{(2),k}\mathbf{X}_{(2),k}^H) + \sigma^2\mathbf{I})^{-1}\mathbf{Y}_{(2),k}
	\label{e14}
\end{multline}
where $\mathbf{X}_{(2),k}\in \mathbb{C}^{Q\times 1}$. 
\paragraph{Estimation of  MU$\rightarrow$IRS-1$\rightarrow$IRS-2$\rightarrow$BS channel}
The LMMSE estimate is given by
\begin{multline}
		\hat{J}_{(3),k} = \\
		 \dfrac{N_2N_1\sigma_1^2\sigma_4^2\sigma_5^2}{\phi_1 \phi_2}\mathbf{X}_{(3),k}^H(N_2N_1\sigma_1^2\sigma_4^2\sigma_5^2(\mathbf{X}_{(3),k}\mathbf{X}_{(3),k}^H) + \sigma^2\mathbf{I})^{-1} \\
	\mathbf{Y}_{(3),k}
	\label{e16}
\end{multline}
where $\mathbf{X}_{(3),k}\in \mathbb{C}^{R\times 1}$.

\subsection{Mean Squared Error (MSE)}

The MSE for the estimate, MU$\rightarrow$IRS-1$\rightarrow$BS is given by
\begin{equation*}
	MSE_{(1),k} = \mathbf{R}_{JJ}-\mathbf{R}_{JY}\mathbf{R}_{YY}^{-1}\mathbf{R}_{YJ}
\end{equation*}
where $\mathbf{R}_{JJ} = N_1\sigma_1^2\sigma_2^2$ is the covariance of $J_{(1),k}$ and $\mathbf{R}_{YJ} = \mathbf{R}_{JY}^H$. Hence, 
\begin{multline}
	MSE_{(1),k}=N_1\sigma_1^2\sigma_2^2- \\
	((N_1\sigma_1^2\sigma_2^2)^2\mathbf{X}_{(1),k}^H(N_1\sigma_1^2\sigma_2^2(\mathbf{X}_{(1),k}\mathbf{X}_{(1),k}^H) 
	+ \sigma^2\mathbf{I})^{-1} \mathbf{X}_{(1),k})
	\label{mse bayesian 1}
\end{multline}
From \eqref{mse bayesian 1}, we can write
\begin{equation}
	MSE_{(1),k} = \left[\dfrac{(\mathbf{X}_{(1),k}^H\mathbf{X}_{(1),k})}{\sigma^2}+\dfrac{1}{N_1\sigma_1^2\sigma_2^2}\right]^{-1}
	\label{mse bayesian1 final}
\end{equation}
\begin{proposition}\label{prop: simplification of MSE}
The $MSE_{(1),k}$ in \eqref{mse bayesian1 final} can be obtained from \eqref{mse bayesian 1}.
\end{proposition}
\begin{proof}
See Appendix-\ref{appendix:proof of prop 2}
\end{proof}
Similarly, the MSE for the estimate, MU$\rightarrow$IRS-2$\rightarrow$BS is
\begin{multline}
	MSE_{(2),k}=N_2\sigma_3^2\sigma_4^2- \\
	((N_2\sigma_3^2\sigma_4^2)^2\mathbf{X}_{(2),k}^H(N_2\sigma_3^2\sigma_4^2(\mathbf{X}_{(2),k}\mathbf{X}_{(2),k}^H) 
	+ \sigma^2\mathbf{I})^{-1} \mathbf{X}_{(2),k})
	\label{mse bayesian 2}
\end{multline}
From \eqref{mse bayesian 2}, we can write
\begin{equation}
	MSE_{(2),k} = \left[\dfrac{(\mathbf{X}_{(2),k}^H\mathbf{X}_{(2),k})}{\sigma^2}+\dfrac{1}{N_2\sigma_3^2\sigma_4^2}\right]^{-1}
	\label{mse bayesian2 final}
\end{equation}
Finally, the MSE for the estimate, MU$\rightarrow$IRS-1$\rightarrow$IRS-2$\rightarrow$BS is
\begin{multline}
	MSE_{(3),k}=N_2N_1\sigma_1^2\sigma_4^2\sigma_5^2- \\
	((N_2N_1\sigma_1^2\sigma_4^2\sigma_5^2)^2\mathbf{X}_{(3),k}^H(N_2N_1\sigma_1^2\sigma_4^2\sigma_5^2(\mathbf{X}_{(3),k}\mathbf{X}_{(3),k}^H) 
	+ \sigma^2\mathbf{I})^{-1} \\
	 \mathbf{X}_{(3),k})
	 \label{mse bayesian 3}
\end{multline}
From \eqref{mse bayesian 3}, we can write
\begin{equation}
	MSE_{(3),k} = \left[\dfrac{(\mathbf{X}_{(3),k}^H\mathbf{X}_{(3),k})}{\sigma^2}+\dfrac{1}{N_2N_1\sigma_1^2\sigma_4^2\sigma_5^2}\right]^{-1}
\end{equation}
\subsection{Cram{\'e}r-Rao lower bound (CRLB)}
Suppose that if there is no prior distribution assigned to the channels then by using the result of CRLB for a complex parameter in \cite{b13}, the CRLB is given by 
\begin{equation}
	CRLB_{(i),k} = \sigma^2 (\mathbf{X}_{(i),k}^H\mathbf{X}_{(i),k})^{-1}
	\label{crlb non bayesian}
\end{equation}
Here, $i\in \{1,2,3\}$ and $k\in\{1,2,..., N\}$.

Since we are considering a Bayesian approach in this paper, from the discussion in \cite[p. 84]{b14} and from \eqref{crlb non bayesian}, we can write the CRLB for the MSE of the estimate, MU$\rightarrow$IRS-i$\rightarrow$BS as
\begin{equation}
	CRLB_{(i),k} = \left[\dfrac{(\mathbf{X}_{(i),k}^H\mathbf{X}_{(i),k})}{\sigma^2}+F_{priori}\right]^{-1}
	\label{crlb bayesian}
\end{equation}
where,
\begin{equation}
 	F_{priori} = \mathbb{E}\left\{\left[\frac{\partial ln(f(J))}{\partial J }\right]^*\left[\frac{\partial ln(f(J))}{\partial J }\right]\right\}
 	\label{priori information}
 \end{equation}
 Here, $f(J)$ is the pdf of the cascaded channel, say $J_{(1),k}$. To calculate $f(J)$, we need to find the pdf of inner product of CSCG random vectors. In \cite{b15}, a joint characteristic function of this problem is derived, which is given as
 \begin{equation}\label{joint cf}
     \psi_{D_1,D_2}(\omega_1,\omega_2) = \dfrac{1}{\left( 1+\sigma_A^2\sigma_B^2\dfrac{\omega_1^2+\omega_2^2}{4}\right)^{S}}
 \end{equation}
 where, $D_1$ and $D_2$ are the real and imaginary parts of $D = \mathbf{A}^H\mathbf{B}$ and $\mathbf{A}\sim\mathcal{CN}(\mathbf{0},\sigma_A^2\mathbf{I})$, $\mathbf{B}\sim\mathcal{CN}(\mathbf{0},\sigma_B^2\mathbf{I})$. Here, $\mathbf{A},\mathbf{B} \in \mathbb{C}^{S\times 1}$. It can be seen from \eqref{joint cf} that it is difficult to find the pdf by applying a inverse fourier transform directly. Hence, by numerical computation, the pdf, $f(J)$ is approximated with Rayleigh pdf. This issue will be further discussed in Section \ref{sim results}. Therefore, 
 \begin{equation}
 F_{priori} = \mathbb{E}\left\{\left|\dfrac{1}{J}-\dfrac{J}{b^2}\right|^2\right\}
\label{numerical priori information}
 \end{equation}
 where $b$ is the scale parameter of Rayleigh pdf and $\dfrac{\partial ln(f(J))}{\partial J} = \left(\dfrac{1}{J}-\dfrac{J}{b^2}\right)$.\\
 
 Now the expectation in \eqref{numerical priori information} can be solved numerically by plugging in the sample values of the inner product ($J$). This numerical computation of the pdf can also be extended to MU$\rightarrow$IRS-1$\rightarrow$IRS-2$\rightarrow$BS channel.
\section{Simulation results}\label{sim results}
The simulations are carried out for the case of MU$\rightarrow$IRS-1$\rightarrow$BS under Bayesian setting. The QPSK symbols are used as pilot symbols in the simulations. Firstly, to compute the Bayesian CRLB, the approximated pdf of $f(J)$ has to be computed to obtain the fisher information. So, we generate $10^5$ samples of the inner product of CSCG random vectors, $\mathbf{H}_{t(1),k}$ and $ \mathbf{H}_{r(1),k}$. We repeat this process for $N_1 = 5$ and $N_1 = 60$ IRS-1 elements. As shown in Figure \ref{fig3}, the pdf is approximated with a Rayleigh distribution having a scale parameter of $b = 7.243$ (estimated). We can observe that this approximated pdf is not a proper fit for the histogram. However, if the number of elements on IRS-1 are increased to $N_1 = 60$, we can observe from Figure \ref{fig4} that the histogram is well approximated with a Rayleigh distribution having $b = 25.1108$ (estimated). This observation is analogous to CLT. After finding out the approximate pdf, the Bayesian CRLB can be computed from \eqref{crlb bayesian}. 
\begin{figure}[!htb]
	\centering
	\includegraphics[height = 5cm, width = 9cm]{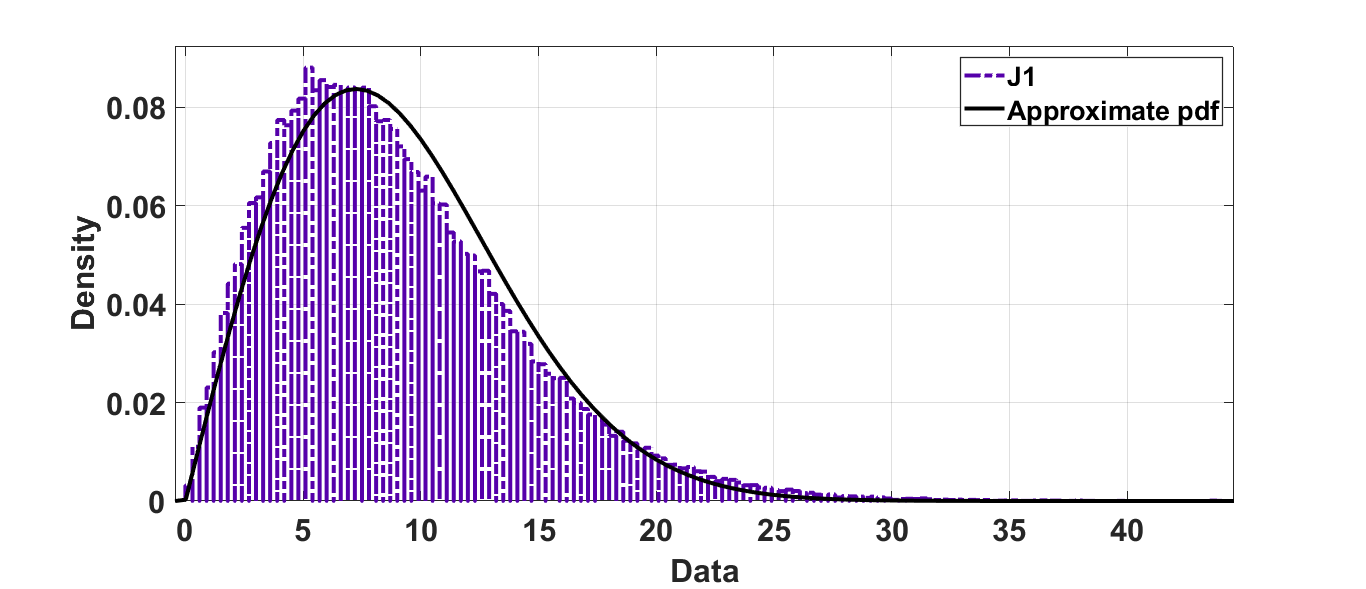}
	\caption{pdf approximation of $J_{(1),k}$ for $N_1=5$}
	\label{fig3}
\end{figure}
\begin{figure}[!htb]
	\centering
	\includegraphics[height = 5cm, width = 9cm]{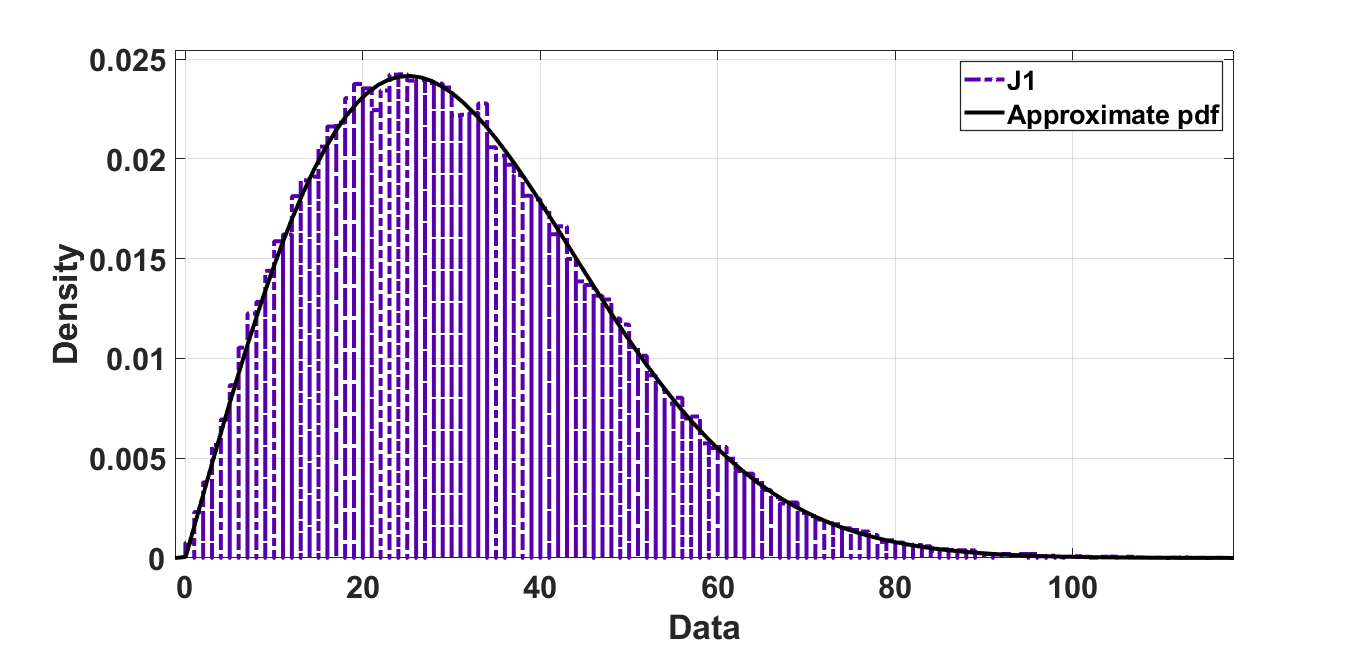}
	\caption{pdf approximation of $J_{(1),k}$ for $N_1=60$}
	\label{fig4}
\end{figure}
 
\begin{figure}[!htb]
	\centering
	\includegraphics[height = 5cm, width = 9cm]{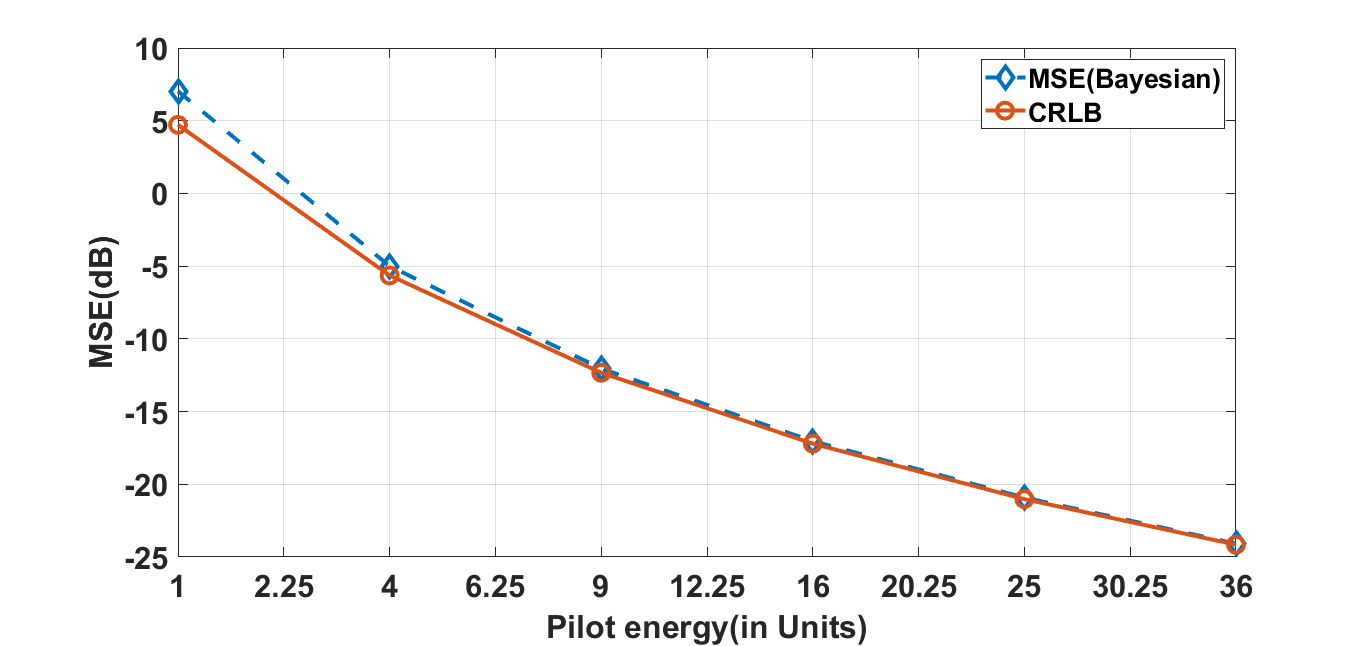}
	\caption{MSE vs Pilot energy for $N_1=60$ in case of Bayesian setting}
	\label{fig5}
\end{figure}
Now we will evaluate the performance of the LMMSE estimator using the Bayesian CRLB which is computed previously. The MSE is used as a performance metric here. It is plotted against the pilot energy in Figure \ref{fig5} by fixing the pilot length to four. Also, the MSE is plotted against pilot length in Figure \ref{fig6} by fixing the pilot energy to unity. The MSE decreases as the pilot energy or pilot length is increased and this result is expected.
\begin{figure}[!htb]
	\centering
	\includegraphics[height = 5cm,width = 9cm]{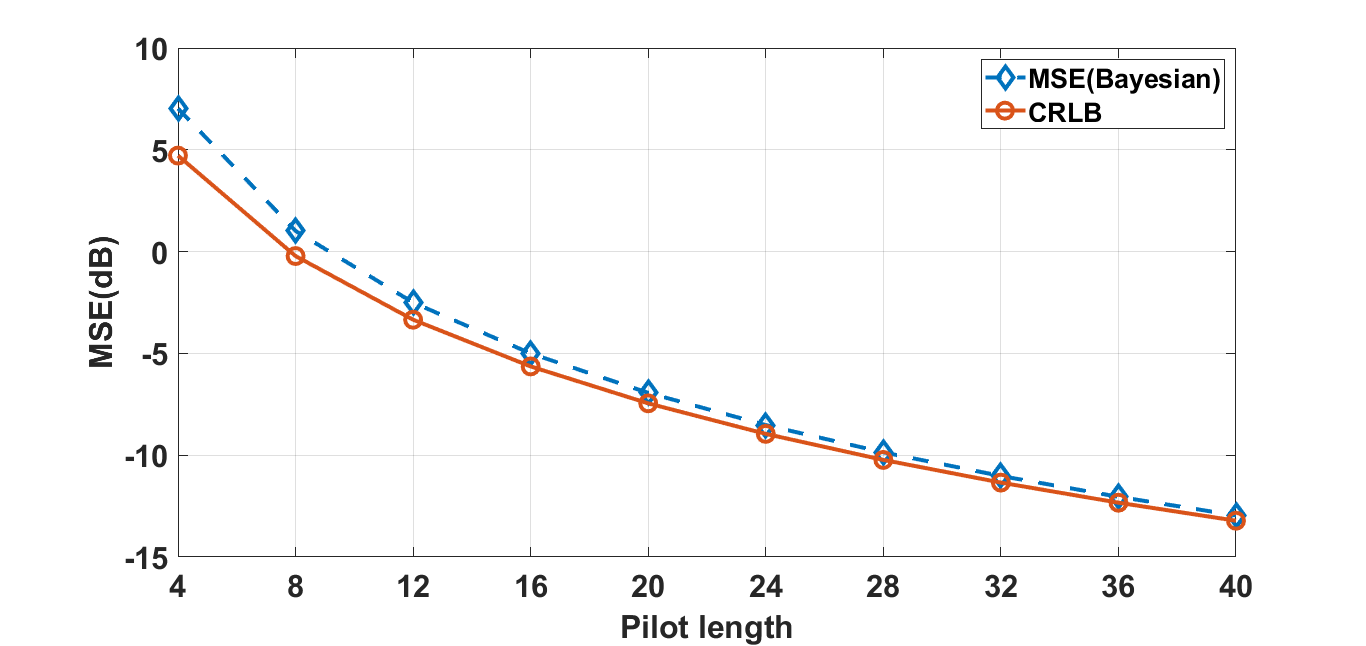}
	\caption{MSE vs Pilot length for $N_1=60$ in case of Bayesian setting}
	\label{fig6}
\end{figure}
 
\section{Conclusion} \label{conclusion}
The cascaded channel estimates are computed in case of double-IRS aided communication with a wideband channel assumption. To evaluate the performance of the estimator (in case of MU$\rightarrow$IRS-1$\rightarrow$BS and MU$\rightarrow$IRS-2$\rightarrow$BS channel links), the Bayesian CRLB is used which was computed numerically. It is shown that computing the CRLB analytically, is difficult because there is no closed form expression of pdf of inner product of CSCG random vectors. Hence, we find the approximate pdf numerically.

\bibliographystyle{IEEEtran} 
\bibliography{references.bib}

\appendices

\section{Proof of Proposition-\ref{prop:mean and variance of J}}\label{appendix:proof of prop 1}
We prove that $J_{(1),k}\sim\mathcal{CN}(0,N_1\sigma_1^2\sigma_2^2)$ where $J_{(1),k} = \mathbf{H}_{t(1),k}^T \mathbf{H}_{r(1),k}$.
Suppose that $\mathbf{a} =  \mathbf{H}_{t(1),k}$ and $\mathbf{b} = \mathbf{H}_{r(1),k}$. Then $J_{(1),k} = \mathbf{a}^T\mathbf{b}$.
We can write $J_{(1),k} = a_1b_1+a_2b_2+...+a_{N_1}b_{N_1}$. Now let $y_1 = a_1b_1$. \\
We can write
\begin{eqnarray}
	y_1 =& (a_{1R}+ja_{1I})(b_{1R}+jb_{1I})\nonumber\\
			=& (a_{1R}b_{1R}-a_{1I}b_{1I})+j(a_{1R}b_{1I}+a_{1I}b_{1R})
			\label{a1}
\end{eqnarray}
From \eqref{a1}, we can write $\mathbb{E}\{y_1\} = \mathbb{E}\{y_{1R}\}+j\mathbb{E}\{y_{1I}\} = 0$ and $Var\{y_1\} = Var\{y_{1R}\}+Var\{y_{1I}\}$.
Consider, 
\begin{eqnarray}
	Var\{a_{1R}b_{1R}\} =&\mathbb{E}\{(a_{1R}b_{1R})^2\}\nonumber\\
														=&Var\{a_{1R}\}Var\{b_{1R}\}\nonumber\\
													=& \dfrac{\sigma_2^2\sigma_1^2}{4}
													\label{a2}
\end{eqnarray}
Hence, from \eqref{a2}, we can write
\begin{eqnarray}
	Var\{y_1\} &=& Var\{a_{1R}b_{1R}\}+Var\{a_{1I}b_{1I}\}+ \nonumber\\
	&&Var\{a_{1R}b_{1I}\}+Var\{a_{1I}b_{1R}\} \nonumber \\
	&=& \sigma_1^2\sigma_2^2
	\label{a3}
\end{eqnarray}
From \eqref{a3}, we can write $Var\{J_{(1),k}\} = N_1\sigma_1^2\sigma_2^2$.\\
Hence, $\mathbb{E}\{J_{(1),k}\} = 0$ and $Var\{J_{(1),k}\} = N_1\sigma_1^2\sigma_2^2$.

\section{Proof of Proposition-\ref{prop: simplification of MSE}}\label{appendix:proof of prop 2}
We show the simplification of $MSE_{(1),k}$. Let $\sigma_h^2 = N_1\sigma_1^2\sigma_2^2$. We ignore the sub-scripts to show the simplification. Consider,
\begin{equation}
	\sigma_h^2\mathbf{X}^H\mathbf{X}\mathbf{X}^H+\sigma^2\mathbf{X}^H
	\label{a12}
\end{equation}
From \eqref{a12}, we can write 
\begin{multline}
	(\sigma_h^2\mathbf{X}^H\mathbf{X}+\sigma^2\mathbf{I})\mathbf{X}^H =
	\mathbf{X}^H(\sigma_h^2\mathbf{X}\mathbf{X}^H+\sigma^2\mathbf{I})
	\label{a13}
\end{multline}
Here, the LHS and RHS in \eqref{a13} originate from \eqref{a12} itself.
By re-arranging the terms in \eqref{a13}, we can write
\begin{equation}
	\mathbf{X}^H(\sigma_h^2\mathbf{X}\mathbf{X}^H+\sigma^2\mathbf{I})^{-1} = 
	(\sigma_h^2\mathbf{X}^H\mathbf{X}+\sigma^2\mathbf{I})^{-1}\mathbf{X}^H
	\label{a14}
\end{equation}
Hence, by substituting \eqref{a14} in \eqref{mse bayesian 1}, we get
\begin{multline*}
	MSE_{(1),k}=N_1\sigma_1^2\sigma_2^2- \\
	((N_1\sigma_1^2\sigma_2^2)^2(N_1\sigma_1^2\sigma_2^2(\mathbf{X}_{(1),k}^H\mathbf{X}_{(1),k}) 
	+ \sigma^2\mathbf{I})^{-1}\mathbf{X}_{(1),k}^H \mathbf{X}_{(1),k})
\end{multline*} 
By adding and subtracting, $\sigma^2\mathbf{I}$, we get
\begin{multline*}
	MSE_{(1),k}=N_1\sigma_1^2\sigma_2^2- \\
	((N_1\sigma_1^2\sigma_2^2)(N_1\sigma_1^2\sigma_2^2(\mathbf{X}_{(1),k}^H\mathbf{X}_{(1),k}) 
	+ \sigma^2\mathbf{I})^{-1}\\ (N_1\sigma_1^2\sigma_2^2(\mathbf{X}_{(1),k}^H\mathbf{X}_{(1),k})
	+\sigma^2\mathbf{I}-\sigma^2\mathbf{I}))
\end{multline*}
Hence,
$$	MSE_{(1),k} = \left[\dfrac{(\mathbf{X}_{(1),k}^H\mathbf{X}_{(1),k})}{\sigma^2}+\dfrac{1}{N_1\sigma_1^2\sigma_2^2}\right]^{-1}$$

\end{document}